\documentclass[conference]{IEEEtran}
\usepackage[utf8]{inputenc}
\usepackage{graphicx}
\usepackage{amssymb}
\usepackage[cmex10]{amsmath}
\usepackage{color}
\usepackage{theorem}
\usepackage{pifont}
\usepackage{array}
\usepackage{cite}
\usepackage{url}
\usepackage{tikz}
\usepackage[english]{babel}
\usepackage{algorithm}
\usepackage{algorithmic}
\usepackage{psfrag}
\usepackage{paralist}

\def\ve#1{{\mathchoice{\mbox{\boldmath$\displaystyle #1$}}%
              {\mbox{\boldmath$\textstyle #1$}}%
              {\mbox{\boldmath$\scriptstyle #1$}}%
              {\mbox{\boldmath$\scriptscriptstyle #1$}}}}
\def\rank#1{\operatorname{rank}(#1)}

\def\Snode{\mathsf{S}}
\def\Dnode{\mathsf{D}}
\def\Nset{\mathcal{N}}
\def\Eset{\mathcal{E}}
\def\Graph{\mathcal{G}}
\def\node{\mathsf{N}}

\def\indeg#1{d_{#1}^{\mathrm{in}}}
\def\outdeg#1{d_{#1}^{\mathrm{out}}}

\def\mincut{\operatorname{mincut}}

\def\T{\mathsf{\!T}}
\def\H{\mathsf{\!H}}
\newcommand{\RM}[1]{\MakeUppercase{\romannumeral #1{}}}

\DeclareSymbolFont{AMSb}{U}{msb}{m}{n}
\DeclareMathSymbol{\C}{\mathalpha}{AMSb}{"43}
\DeclareMathSymbol{\F}{\mathalpha}{AMSb}{"46}
\DeclareMathSymbol{\R}{\mathalpha}{AMSb}{"52}
\DeclareMathSymbol{\N}{\mathalpha}{AMSb}{"4E}

\newtheorem{myTheo}{Theorem}

\begin{document}
\title{Layering of Communication Networks\\and a Forward-Backward Duality\vspace*{-10mm}}

\IEEEoverridecommandlockouts
\author{\IEEEauthorblockN{%
Michael Cyran${}^1$, Birgit Schotsch${}^2$, 
Johannes B. Huber${}^1$, Robert F.H. Fischer${}^3$, Vahid Forutan${}^3$%
}%
\IEEEauthorblockA{%
${}^1$Lehrstuhl f{\"u}r Informations{\"u}bertragung, 
Friedrich-Alexander-Universit{\"a}t Erlangen-N{\"u}rnberg, Erlangen, Germany\\
${}^2$Airbus DS GmbH, Munich, Germany\\
${}^3$Institut f{\"u}r Nachrichtentechnik, Universit{\"a}t Ulm, Ulm, Germany\\
\fontsize{8pt}{10pt plus 0.2pt minus 0.1pt}\tt michael.cyran@fau.de, 
birgit.schotsch@airbus.com, johannes.huber@fau.de,\\robert.fischer@uni-ulm.de, 
vahid.forutan@uni-ulm.de%
\thanks{This work was supported by DFG under grants FI~982/4-3 and HU~634/11-3, 
        and by BMBF under grant 16~BP~12406.}
}
\vspace*{-3mm}
}

\maketitle

\begin{abstract}
In layered communication networks there are only connections between 
intermediate nodes in adjacent layers. Applying network coding to such networks 
provides a number of benefits in theory as well as in practice. We propose a 
\emph{layering procedure} to transform an arbitrary network into a layered 
structure. Furthermore, we derive a \emph{forward-backward duality} for linear 
network codes, which can be seen as an analogon to the \emph{uplink-downlink 
duality} in MIMO communication systems.
\end{abstract}

\vspace*{-2mm}
\section{Introduction}
\noindent In \cite{ahlswede2000network} it was shown that communication between 
two nodes within a communication network is possible up to a rate that is equal 
to the minimum rate flowing through any possible cut between these two 
nodes---the \emph{mincut} between them. This rate can be achieved by allowing 
intermediate nodes to \emph{code}, i.e., to calculate functions of their 
incoming messages before forwarding them. In \cite{li2003linear} it was proved 
that it suffices to apply \emph{linear network coding} (LNC), i.e., 
intermediate nodes just need to form \emph{linear} combinations of their 
received messages from a finite field $\F_q$. If all operations are performed 
over a finite field of large enough size $q$, the factors at the intermediate 
nodes may even be drawn independently at random, which leads to a robust, 
decentralized, and capacity achieving approach: \emph{random linear network 
coding} (RLNC) \cite{ho2006,silva2010}.

This paper studies network coding (NC) in \emph{layered networks}, where 
intermediate nodes are arranged in layers and there exist only edges between 
nodes which are located in adjacent layers. We introduce a \emph{layering}
procedure for establishing a layered structure in seemingly disparate and 
unstructured network topologies. Applying NC to a layered network provides a 
number of benefits in theory for analysis as well as in practice. Moreover, we 
address the problem of \emph{bidirectional NC} and derive a 
\emph{forward-backward duality}. 

The paper is organized as follows: Sec.~\ref{Sec:NC} gives a brief 
recapitulation and a classification of NC. In Sec.~\ref{Sec:layering} we 
examine layered networks and introduce the \emph{layering procedure}. 
\emph{Bidirectional NC} is discussed in Sec.~\ref{Sec:BidirectionalNC} and 
some conclusions are drawn in Sec.~\ref{Sec:Conclusion}.

\vspace*{-2mm}
\section{Brief Recapitulation of Network Coding} \label{Sec:NC}
\subsection{Problem Formulation}
\vspace*{-1mm}
We define a communication network as a directed, acyclic graph 
$\Graph=\{\Nset,\Eset\}$ with a set of nodes $\Nset$ and a set of edges 
$\Eset$. The considered \emph{multicast scenario} consists of a unique source 
node $\Snode \in \Nset$ with $n$ outgoing edges, and $K$ destination nodes 
$\Dnode_k$, $k=1,\dots,K$, with $N_k \ge n$ incoming edges. The source 
transmits $n$ symbols $x_1,\dots,x_n \in \F_q$ to each of the destination nodes 
$\Dnode_k$ by injecting these $n$ symbols in parallel (one on each of its 
outgoing edges) into the network and each destination node $\Dnode_k$ tries to 
reconstruct all these symbols from its $N_k$ receive symbols 
$y_{k,1},\dots,y_{k,N_k} \in \F_q$. Nodes within the network are connected by 
edges $e_{i,j} = (\node_i, \node_j) \in \Eset$.
Each edge represents a noiseless%
\footnote{Since we do not treat error-correction coding for networks in this
paper, we restrict ourselves to the case of error-free NC. However, all 
statements contained in this paper are also applicable for noisy networks.}%
communication link on which one symbol from $\F_q$ can be transmitted per 
usage. We further assume that each edge induces the same delay.%
\footnote{If this is not the case, equal-delay edges can be achieved through
          appropriate buffers at the intermediate nodes.}
The in-degree $\indeg{i}$ and the out-degree $\outdeg{i}$ of a node $\node_i$ 
is defined as the number of its incoming and outgoing edges, respectively. 
Coding at intermediate nodes is accomplished as follows: each node $\node_i$ 
collects the symbols from each of its $\indeg{i}$ incoming edges. Then, it 
computes possibly different functions of these symbols and transmits them on 
its $\outdeg{i}$ outgoing edges. 
\vspace*{-2mm}
\subsection{Classification of Network Coding Variants}
Essentially, there exist two distinct approaches to generate outgoing messages
at intermediate nodes. In the first one, which we denote as NC Variant~\RM{1},
each intermediate node calculates only a single function of its input symbols
and transmits the resulting output symbol on all outgoing edges.
%
%
This variant is applicable, e.g., in wireless networks, where intermediate
nodes possess omnidirectional antennas, and thus, transmit a single signal. In
NC Variant~\RM{2} intermediate nodes compute \emph{individual} output symbols 
for their outgoing edges. This variant can be applied, e.g., in wired networks.
 \begin{figure}
  \centering
  \psfrag{p1}[c][c][1][0]{$z_{i,1}^{\mathrm{in}}$}
  \psfrag{pr}[c][c][1][0]{$z_{i,{\indeg{i}}}^\mathrm{in}$}
  \psfrag{po1}[c][c][1][0]{$z_{i,1}^{\mathrm{out}}$}
  \psfrag{por}[c][c][1][0]{$z_{i,{\outdeg{i}}}^{\mathrm{out}}$}
  \psfrag{ddd}[c][c][1][90]{$\dots$}
  \psfrag{N}[cc][bl][.8][0]{$\node_i$}
  \psfrag{N1}[cc][bl][.8][0]{$\node_{i,1}$}
  \psfrag{N2}[cc][bl][.8][0]{$\node_{i, \outdeg{i}}$}
  \psfrag{a}[c][c][1][0]{(a)}
  \psfrag{b}[c][c][1][0]{(b)}
  \includegraphics[width=.4\textwidth]{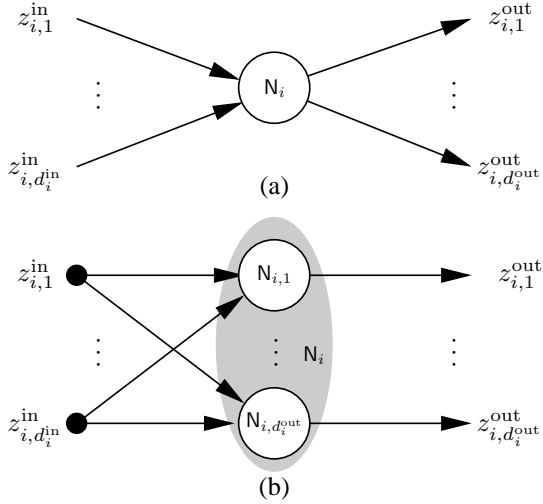}
  \caption{Illustration of Theorem~\ref{Theo:var12}: Conversion of a node 
           $\node_i$ which applies NC Variant~\RM{2} (a) into $\outdeg{i}$ 
           single output nodes $\node_{i,1},\dots,\node_{i,\outdeg{i}}$ (b).}
  \label{Fig:var1var2conv}
 \end{figure}
In Fig.~\ref{Fig:var1var2conv}(a) an intermediate node $\node_i$ with 
$\indeg{i}$ incoming and $\outdeg{i}$ outgoing edges is depicted. The incoming 
and the outgoing symbols of node $\node_i$ are denoted as 
$z_{i,\delta}^{\mathrm{in}}$, $\delta=1,\dots,\indeg{i}$, and 
$z_{i,\rho}^{\mathrm{out}}$, $\rho=1,\dots,\outdeg{i}$, respectively. The two 
NC variants are closely related to each other. This is specified in the 
following theorem and is illustrated in Fig.~\ref{Fig:var1var2conv}.
\vspace*{-2mm}
\begin{myTheo} \label{Theo:var12}
 A communication network employing NC Variant~\RM{2} can be transformed into an 
 equivalent network which applies NC Variant~\RM{1}, by splitting up each 
 intermediate node $\node_i$ with $\outdeg{i}$ outgoing edges into $\outdeg{i}$ 
 single output auxiliary nodes. These auxiliary nodes possess the same input 
 edges as the original node $\node_i$.
\end{myTheo}
\vspace*{-1mm}
\begin{proof}
 A Variant-\RM{2} node $\node_i$ is split up into $\outdeg{i}$ auxiliary single 
 output nodes $\node_{i,j}$, $j=1,\dots,\outdeg{i}$, cf. 
 Fig.~\ref{Fig:var1var2conv}(b).
 By repeating this procedure for all Variant-\RM{2} nodes results in an 
 equivalent NC Variant~\RM{1} network.
\end{proof}

Hybrid forms of these two variants are also possible, if a node $\node_i$ 
transmits $h < \outdeg{i}$ distinct messages. 
Such a variant is possible, e.g., in wireless networks, where intermediate 
nodes possess several directional antennas and transmit distinct messages in 
distinct directions. These hybrid variants can also be transformed into NC 
Variant~\RM{1} by splitting up nodes which transmit $h$ different messages 
into $h$ auxiliary nodes.

Obviously, the mincut of a network can only be achieved by applying NC 
Variant~\RM{2}. However, for analysis the equivalent NC Variant~\RM{1} 
representation is more convenient, as will be shown in the remainder of this 
paper. 
\vspace*{-2mm}
\subsection{Linear Network Coding}
In LNC the outgoing messages $z_{i,\rho}^{\mathrm{out}}$ at a node $\node_i$ 
are $\F_q$-linear combinations of their incoming messages 
$z_{i,\delta}^{\mathrm{in}}$
\begin{equation} \label{Eq:linNC}
 z_{i,\rho}^{\mathrm{out}} = \sum_{\delta=1}^{\indeg{i}} c_{i,\delta,\rho}^{}
                         \cdot z_{i,\delta}^{\mathrm{in}} \, ,
 \quad \rho = 1,\dots,\outdeg{i} \, ,
\end{equation}
where $c_{i,\delta,\rho} \in \F_q$ are the \emph{linear coding coefficients} at 
node $\node_i$. If NC Variant~\RM{1} is applied, all outgoing symbols are 
equal, i.e., 
$z_{i}^{\mathrm{out}}=z_{i,1}^{\mathrm{out}}=\ldots=
z_{i,\outdeg{i}}^{\mathrm{out}}$, and thus, $c_{i,\delta,\rho}=c_{i,\delta}$, 
$\forall \rho$, whereas in NC Variant~\RM{2} these quantities are different.

Let $\ve{z}_i^{\mathrm{in}} \in \F_q^{\indeg{i}}$ and 
$\ve{z}_i^{\mathrm{out}} \in \F_q^{\outdeg{i}}$ be the vectors of incoming and
outgoing symbols at node $\node_i$, respectively. We can write (\ref{Eq:linNC}) 
in vector-matrix notation as
\begin{equation}
 \ve{z}_i^{\mathrm{out}} = \ve{C}_i^{} \, \ve{z}_i^{\mathrm{in}} \, ,
\end{equation}
where $\ve{C}_i \in \F_q^{\outdeg{i} \times \indeg{i}}$ is the 
\emph{coefficient matrix} of node $\node_i$
\begin{eqnarray} \label{Eq:CodingMatrix}
 \ve{C}_{i} &=& 
 \begin{bmatrix}
  c_{i,1,1} & \cdots & c_{i,\indeg{i},1} \\
  c_{i,1,2} & \cdots & c_{i,\indeg{i},2} \\
  \vdots       & \ddots & \vdots \\
  c_{i,1,\outdeg{i}} & \cdots & c_{i,\indeg{i},\outdeg{i}}
 \end{bmatrix} \, ,
\end{eqnarray}
%
In NC Variant~\RM{1} the columns of $\ve{C}_i$ are restricted to one element 
($c_{i,\delta,\rho} = c_{i,\delta}$, $\forall \rho$), whereas in NC 
Variant~\RM{2} the columns consist of individual entries.

Since each intermediate node performs linear coding, the resulting receive 
vector $\ve{y}_k = \begin{bmatrix} y_{k,1}, \ldots,\, y_{k,N_k} 
\end{bmatrix}^\mathsf{T} \in \F_q^{N_k}$ is still a linear transformation of 
the source vector $\ve{x} = \begin{bmatrix} x_1, \ldots,\, x_n 
\end{bmatrix}^\mathsf{T} \in \F_q^n$, i.e., the network between source $\Snode$ 
and destination $\Dnode_k$ acts as a linear map $\F_q^n \rightarrow \F_q^{N_k}$ 
which is represented by the \emph{individual network channel matrix} $\ve{A}_k 
\in \F_q^{N_k \times n}$. The elements %
$a_{i,j}$ of this matrix represent the corresponding \emph{route gains}, i.e., 
$a_{i,j}$ is the gain of the route from the $j$th outgoing edge of the source 
node $\Snode$ to the $i$th incoming edge of destination node $\Dnode_k$. These 
route gains are sums of products of the coding coefficients 
$c_{i,\delta,\rho}$. The end-to-end model for a $\Snode \rightarrow 
\Dnode_k$ link is given by
\begin{equation} \label{Eq:MMC}
 \ve{y}_k = \ve{A}_k^{} \ve{x} \, .
\end{equation}
$\Dnode_k$ is able to reconstruct $\ve{x}$ if $\ve{A}_k$ has full column
rank $n$. We speak of a \emph{valid NC} in this case.

\vspace*{-2mm}
\section{Layering} \label{Sec:layering}
\vspace*{-2mm}
\subsection{Layered Networks}
\vspace*{-1mm}
In a layered network all intermediate nodes are arranged in $L$ layers. 
Nodes in layer $l$ only receive packets from nodes in layer $l-1$, i.e., there 
are no connections between non-adjacent layers and no connections between 
nodes within the same layer. In Fig. \ref{Fig:layeredNW} a layered network with
one source node $\Snode$ and $K$ destination nodes $\Dnode_k$, $k=1,\dots,K$, 
is depicted.
\begin{figure}
 \centering
 \psfrag{S}[cc][bl][1][0]{$\Snode$}
 \psfrag{D}[cc][bl][1][0]{$\Dnode_1$}
 \psfrag{DK}[cc][bl][1][0]{\hspace*{1mm}$\Dnode_K$}
 \psfrag{DOTS}[cc][bl][1][0]{$\cdots$}
 \psfrag{X1}[bc][bl][1][0]{$x_1$}
 \psfrag{X2}[cl][tl][1][0]{$x_2$}
 \psfrag{X3}[cr][tl][1][0]{$x_3$}
 \psfrag{Xn}[bl][bl][1][0]{$x_n$}
 \psfrag{Y1}[cc][cc][1][0]{$y_{1,1}$}
 \psfrag{Y2}[c][c][1][0]{$y_{1,N_1}$}
 \psfrag{Y3}[cc][cc][1][0]{$y_{K,1}$}
 \psfrag{YN1}[cc][cc][1][0]{$y_{K,N_K}$}
 \psfrag{A1}[cl][bl][1][0]{$\ve{A}_{2,1}$}
 \psfrag{A2}[cl][bl][1][0]{$\ve{A}_{3,2}$}
 \psfrag{AL-1}[cl][bl][1][0]{$\ve{A}_{L-1,L-2}$}
 \psfrag{AL}[cl][bl][1][0]{$\ve{A}_{L,L-1}$}
 \psfrag{L0}[cl][bl][1][0]{$l=1$}
 \psfrag{L1}[cl][bl][1][0]{$l=2$}
 \psfrag{L2}[cl][bl][1][0]{$l=3$}
 \psfrag{LL2}[cl][bl][1][0]{$l=L-2$}
 \psfrag{LL1}[cl][bl][1][0]{$l=L-1$}
 \psfrag{LL}[cl][bl][1][0]{$l=L$}
 \psfrag{d}[cl][bl][1][0]{$\cdots$}
 \includegraphics{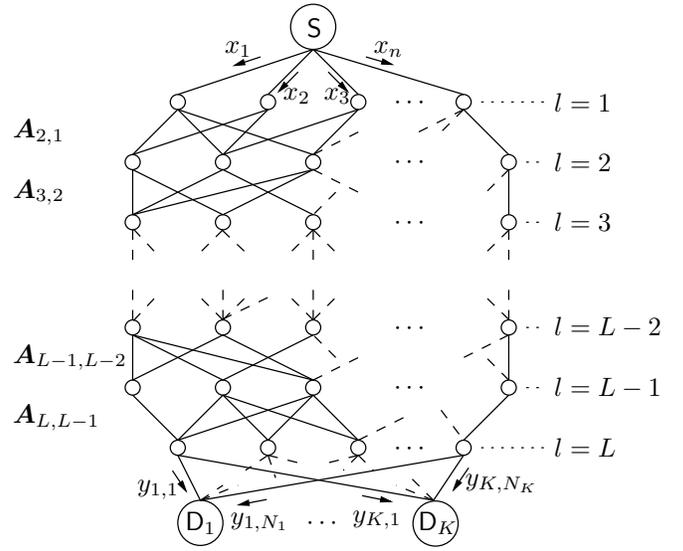}
 \caption{\label{Fig:layeredNW}Exemplary layered network with $L$ layers, one 
          source node $\Snode$, and $K$ destination nodes $\Dnode_k$ 
          (multicast scenario).}
\end{figure}
The number of nodes in layer $l$ is denoted as $n_l$, with $n_1 = n$ and 
$n_L \ge \max_k(N_k)$. For the unicast scenario, i.e., if there is only one 
destination node $\Dnode$, $n_L = N$ holds. 

Such networks exhibit a number of beneficial properties of which two are 
particularly noteworthy. 
\begin{enumerate}
 \item A layered network is inherently time synchronized. All symbols arrive 
       simultaneously at a specific intermediate node. Consequently, each 
       intermediate node can immediately code its incoming symbols and does not 
       have to wait until all required symbols arrive.
 \item It enables a factorization of the individual network channel matrices 
       $\ve{A}_k$ (cf. Sec.~\ref{Sec:LNCinLayNets}). This is the basis for the 
       derivation of the \emph{forward-backward duality} for LNC (cf. 
       Sec.~\ref{Sec:BidirectionalNC}). 
\end{enumerate}
\subsection{Linear Network Coding in Layered Networks} \label{Sec:LNCinLayNets}
When \emph{linear} NC Variant~\RM{1} is applied,%
\footnote{If the network nodes apply NC Variant~\RM{2}, the network can be 
          transformed into an equivalent network which applies NC 
          Variant~\RM{1} (cf., Theorem~\ref{Theo:var12}), and the factorization 
          of the channel matrix has to be accomplished for the equivalent 
          network.} %
the \emph{overall network channel matrix} $\ve{A}$, i.e., the linear 
transformation from layer 1 to layer $L$, can be obtained as the product of all 
$L-1$ \emph{interlayer matrices} $\ve{A}_{l+1,l} \in \F_q^{n_{l+1} \times n_l}$
\begin{equation} \label{Eq:Aprod}
 \ve{A} = \ve{A}_{L,L-1} \cdot \ve{A}_{L-1,L-2} \cdots 
          \ve{A}_{2,1} = \prod_{l=1}^{L-1} \ve{A}_{l+1,l} \, .
\end{equation}
These interlayer matrices consist of the linear factors associated with the
edges that connect the corresponding layers. The element in the $i$th row and
the $j$th column of $\ve{A}_{l+1,l}$ represents the linear factor corresponding
to the edge which connects the $j$th node in layer $l$ with the $i$th node in 
layer $l+1$. The connection between the interlayer matrices and the coefficient 
matrices is as follows. $\ve{A}_{l+1,l}$ contains the coding coefficients of 
the coefficient matrices $\ve{C}_i$ which correspond to the intermediate nodes 
in layer $l+1$. In addition to that, the interlayer matrices imply the wiring 
between the two affected layers, whereas the coefficient matrices merely 
describe the operations at one specific node. To sum up, $\ve{A}_{l+1,l}$ is an
\emph{edge-oriented} description of the LNC, which takes also the topology into
account, and $\ve{C}_i$ is a local, \emph{node-oriented} description.

The \emph{individual} network channel matrix $\ve{A}_k$ corresponding to 
destination node $\Dnode_k$, $k=1,\dots,K$, consists of a subset 
$\mathcal{D}_k$ of rows%
\footnote{We adopt the Matlab notation, i.e., 
$\ve{A}(\mathcal{A}, \mathcal{B})$, represents a matrix composed of a subset 
$\mathcal{A}$ of the rows and a subset $\mathcal{B}$ of the columns of 
$\ve{A}$.} %
of $\ve{A}$
\begin{equation}
 \ve{A}_k = \ve{A}(\mathcal{D}_k,:) \, ,
\end{equation}
where $\mathcal{D}_k$ is the subset of rows, which correspond to the nodes in 
the last layer, to which the destination node $\Dnode_k$ is connected. In case 
of the unicast scenario, the individual network channel matrix is equal to the 
overall network channel matrix $\ve{A}$. 

The factorization (\ref{Eq:Aprod}) enables a simple method to determine an 
upper bound on the mincut between the source and a destination:
\begin{myTheo} \label{Theo:MinCutLinNC}
 The mincut between the source $\Snode$ and a destination node $\Dnode_k$ in a 
 layered network is
 \begin{eqnarray}
  \mincut(\Snode,\Dnode_k) &=& 
   \max_{{c_{i,\delta,\rho} \in \F_q} \atop {q > K}} \, (\rank{\ve{A}_k})  
\nonumber\\
   &\le& \min_l \, (\, \max_{{c_{i,\delta,\rho} \in \F_q} \atop {q > K}}
     \, (\rank{\ve{A}_{l+1,l}})) \, .
 \end{eqnarray}
\end{myTheo}
\begin{proof}
 The mincut between $\Snode$ and $\Dnode_k$ is the number of symbols which can 
 be reliably transmitted from $\Snode$ to $\Dnode_k$, and thus, is equal to the 
 rank of the individual network channel matrix. Since the individual network 
 channel matrix is the product of the corresponding inter-layer matrices, the 
 minimal rank of the inter-layer matrices is an upper bound on the mincut 
 between $\Snode$ and $\Dnode_k$. The finite field size $q$ has to be greater 
 than the number of destinations $K$ \cite{jaggy2005}.
\end{proof}
\subsection{Layering of Arbitrary Networks} \label{Subsec:layering}
In a non-layered network paths from the source node to the destination nodes 
consist of different numbers of edges i.e., have different ``lengths''. An 
exemplary non-layered network is depicted in Fig.~\ref{Fig:nonlayeredNW}(a).
\begin{figure}
 \centering
 \psfrag{S}[cc][bl][1][0]{$\Snode$}
 \psfrag{D1}[cc][bl][1][0]{$\Dnode_1$}
 \psfrag{D2}[cc][bl][1][0]{$\Dnode_2$}
 \psfrag{N1}[cc][bl][1][0]{$\node_1$}
 \psfrag{N2}[cc][bl][1][0]{$\node_2$}
 \psfrag{N3}[cc][bl][1][0]{$\node_3$}
 \psfrag{N4}[cc][bl][1][0]{$\node_4$}
 \psfrag{N5}[cc][bl][1][0]{$\node_5$}
 \psfrag{z}[cc][bl][1][0]{$\mathrm{D}$}
 \psfrag{z2}[cc][bl][1][0]{$\mathrm{D}^2$}
 \psfrag{a}[cc][bl][1][0]{(a)}
 \psfrag{b}[cc][bl][1][0]{(b)}
 \includegraphics{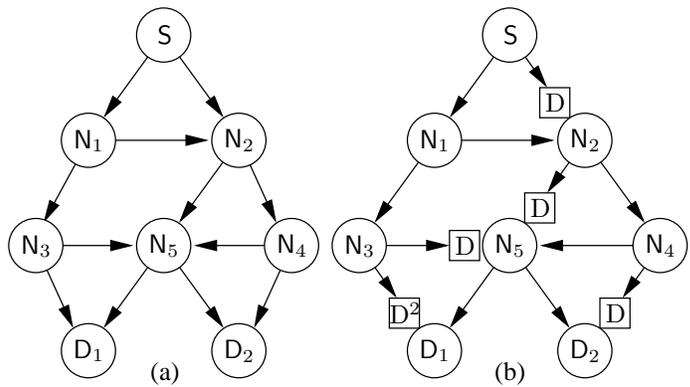}
 \caption{\label{Fig:nonlayeredNW}Non-layered network with one source, two 
          destinations, and five intermediate nodes. Without (a), and with 
          depicted delay elements (b).} 
\end{figure}
Obviously, there are paths from $\Snode$ to $\Dnode_k$ $(k=1,2)$ of different 
lengths, e.g., 
$\Snode \rightarrow \node_1 \rightarrow \node_3 \rightarrow \Dnode_1$ and 
$\Snode \rightarrow \node_1 \rightarrow \node_2  \rightarrow \node_4 
\rightarrow \node_5 \rightarrow \Dnode_1$ consisting of three and five edges, 
respectively. The aim of our proposed procedure, which we denote as 
\emph{layering}, is to force all paths from the source to all of the 
destinations to have the same length, namely $L+1$. 

For that, consider the \emph{coding points}, i.e., the nodes which receive more 
than one symbol. The first coding point in our exemplary network in 
Fig.~\ref{Fig:nonlayeredNW}(a) is $\node_2$, which receives a packet from 
$\Snode$ after one time unit, and a packet from $\node_1$ after two time units. 
To be able to code, i.e., to create a function of these two packets, $\node_2$ 
has to buffer the packet received from $\Snode$ for one time unit. This buffer, 
which actually is part of $\node_2$, can formally be redrawn outside of 
$\node_2$. We continue this step for all coding points in 
$\Graph\{\Nset,\Eset\}$ and obtain the network depicted in 
Fig.~\ref{Fig:nonlayeredNW}(b). A delay of $s$ time units is denoted as 
$\mathrm{D}^s$. Finally, we interpret these delay elements as 
single-input/single-output (SISO) nodes, which just pass the packet received on 
their incoming edge to their outgoing edge. Delays of $s$ time units are 
interpreted as $s$ consecutive SISO nodes. Basically, layering consists of two 
steps:
\begin{enumerate}
 \item Enumerate all intermediate network nodes according to an ancestral 
       ordering%
       \footnote{Such an ordering exists for all acyclic networks 
                 \cite{KoM03}.}%
       , i.e., if $e_{i,j} \in \Eset$ then $i<j$. 
 \item Visit all coding points sequentially and introduce SISO nodes, such that 
       all paths which meet in one point have the same length.
\end{enumerate}
\begin{figure}
 \centering
 \psfrag{S}[cc][bl][1][0]{$\Snode$}
 \psfrag{D1}[cc][bl][1][0]{$\Dnode_1$}
 \psfrag{D2}[cc][bl][1][0]{$\Dnode_2$}
 \psfrag{N1}[cc][bl][1][0]{$\node_1$}
 \psfrag{N2}[cc][bl][1][0]{$\node_2$}
 \psfrag{N3}[cc][bl][1][0]{$\node_3$}
 \psfrag{N4}[cc][bl][1][0]{$\node_4$}
 \psfrag{N5}[cc][bl][1][0]{$\node_5$}
 \psfrag{l1}[cc][bl][1][0]{$l=1$}
 \psfrag{l2}[cc][bl][1][0]{$l=2$}
 \psfrag{l3}[cc][bl][1][0]{$l=3$}
 \psfrag{l4}[cc][bl][1][0]{$l=4$}
 \psfrag{A1}[cc][cc][1][0]{$\ve{A}_{2,1} \in \F_q^{2 \times 2}$}
 \psfrag{A2}[cc][cc][1][0]{$\ve{A}_{3,2} \in \F_q^{4 \times 2}$}
 \psfrag{A3}[cc][cc][1][0]{$\ve{A}_{4,3} \in \F_q^{3 \times 4}$}
 \includegraphics{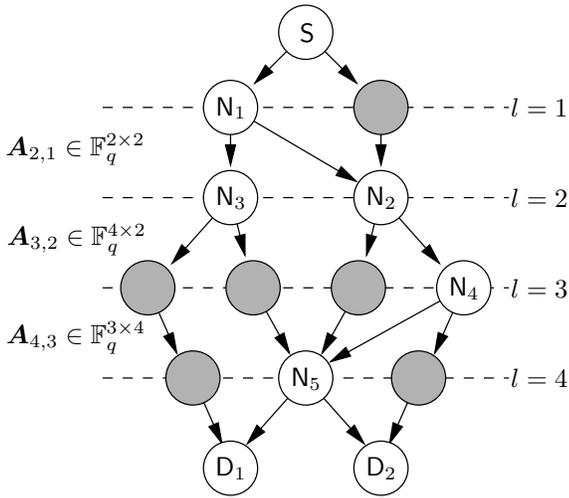}
 \caption{\label{Fig:layering}Communication network from 
          Fig.~\ref{Fig:nonlayeredNW} in layered representation.}
\end{figure}
After redrawing the network, we obtain the layered structure depicted in 
Fig.~\ref{Fig:layering}, where the introduced SISO nodes are depicted in gray. 
This layered network with $L=4$ layers is equivalent to the network depicted in 
Fig.~\ref{Fig:nonlayeredNW}(a). Since each coding point has to be visited 
exactly once, the complexity of this algorithm is of order 
$\mathcal{O}(N_\mathrm{cp}^{} \cdot \bar{d}_\mathrm{cp}^{\mathrm{in}})$, where 
$N_\mathrm{cp}$ is the number of coding points and 
$\bar{d}_\mathrm{cp}^{\mathrm{in}}$ is the average number of incoming edges of 
the coding points. We summarize this insight in the following theorem.
\begin{myTheo} \label{Theo:layering}
 Despite the actual structure of an acyclic network, an equivalent layered 
 network can be obtained by introducing additional redundant SISO nodes, such 
 that all paths from the source to any destination consist of the same number 
 of edges.
\end{myTheo}

The factorization (\ref{Eq:Aprod}) of $\ve{A}$ can be accomplished together 
with the layering procedure: During the layering procedure the nodes are 
assigned to layers and the wiring between the layers can be obtained from the 
set of edges $\Eset$.

We speak of a \emph{layered Variant~\RM{1} representation} of an arbitrary 
network if it was layered according to Theorem~\ref{Theo:layering} and 
transformed to Variant~\RM{1} according to Theorem~\ref{Theo:var12}. In 
\cite{schotsch2015pout} we already exploited the layered Variant~\RM{1}
representation of communication networks in the context of RLNC. With the aid 
of the factorized version of the network channel matrix (\ref{Eq:Aprod}) we 
derived in \cite{schotsch2015pout} the probability distribution of the entries 
of $\ve{A}$ and an upper bound on the outage probability of random linear 
network codes with known incidence matrices. A further consequence of the 
layered Variant~\RM{1} representation is a new possibility of the determination 
of an upper bound on the mincut of acyclic networks in two steps:
\begin{enumerate}
 \item Layering of the network and a Variant~\RM{2} to Variant~\RM{1} 
       conversion if necessary.
 \item Determination of the mincut according to Theorem \ref{Theo:MinCutLinNC}.
\end{enumerate}

\vspace*{-2mm}
\section{Bidirectional Network Coding} \label{Sec:BidirectionalNC}
\noindent Up to now, we have considered a unidirectional communication from the 
source node $\Snode$ to one or several destination nodes $\Dnode_k$. In this 
section, we address the problem of a bidirectional communication between a 
source-destination pair, i.e., the case where a destination node $\Dnode_k$ 
replies to the source node $\Snode$, which is of interest, e.g., in optical 
(fiber-optical) networks. For the moment, we assume that $N_k = n$, i.e., that 
the individual network channel matrix $\ve{A}_k$ is square. Furthermore, for 
notational convenience, we drop the index $k$ and denote the considered 
individual network channel matrix as $\ve{A}$.

When we reverse the direction of communication, it is reasonable to reverse the 
operations at the intermediate nodes, as depicted in Fig.~\ref{Fig:revert} for
the case of a node with two incoming and two outgoing edges. In the backward
direction, not only the direction of communication is reversed, also the 
summing and the distribution points are interchanged.
\begin{figure}
 \centering
 \psfrag{a1}[cc][bl][1][0]{$c_{1,1}$}
 \psfrag{a2}[cc][bl][1][0]{$c_{2,1}$}
 \psfrag{a3}[cc][bl][1][0]{$c_{1,2}$}
 \psfrag{a4}[cc][bl][1][0]{$c_{2,2}$}
 \psfrag{zin1}[c][c][1][0]{$z_1^\mathrm{in}$}
 \psfrag{zin2}[c][c][1][0]{$z_2^\mathrm{in}$}
 \psfrag{zout1}[c][c][1][0]{$z_1^{\mathrm{out}}$}
 \psfrag{zout2}[c][c][1][0]{$z_2^{\mathrm{out}}$}
 \psfrag{a}[c][c][1][0]{(a)}
 \psfrag{b}[c][c][1][0]{(b)}
 \psfrag{zin1b}[c][c][1][0]{$z_{\mathrm{b},1}^\mathrm{in}$}
 \psfrag{zin2b}[c][c][1][0]{$z_{\mathrm{b},2}^\mathrm{in}$}
 \psfrag{zout1b}[c][c][1][0]{$z_{\mathrm{b},1}^{\mathrm{out}}$}
 \psfrag{zout2b}[c][c][1][0]{$z_{\mathrm{b},2}^{\mathrm{out}}$}
 \includegraphics[width=.8\columnwidth]{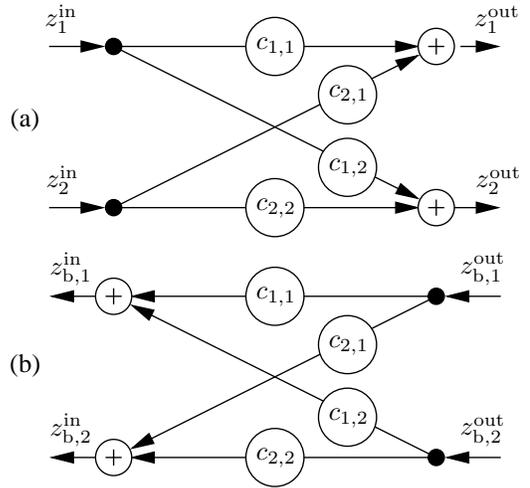}
 \caption{Exemplary intermediate node with two incoming and two outgoing edges 
      in forward (a) and backward (b) direction.}
 \label{Fig:revert}
\end{figure}
The input-output relation of this exemplary node by means of the coefficient 
matrices (\ref{Eq:CodingMatrix}) for the \emph{forward direction} is
\begin{equation}
 \begin{bmatrix}
  z_1^\mathrm{out} \\ z_2^\mathrm{out}
 \end{bmatrix} = 
 \underbrace{%
 \begin{bmatrix}
  c_{1,1} & c_{2,1} \\
  c_{1,2} & c_{2,2}
 \end{bmatrix}
 }_{\ve{C}_i}
 \begin{bmatrix}
  z_1^{\mathrm{in}} \\ z_2^{\mathrm{in}}
 \end{bmatrix} \, ,
\end{equation}
whereas for the \emph{backward direction} we obtain
\begin{equation}
 \begin{bmatrix}
  z_{\mathrm{b},1}^\mathrm{out} \\ z_{\mathrm{b},2}^\mathrm{out}
 \end{bmatrix} = 
 \underbrace{%
 \begin{bmatrix}
  c_{1,1} & c_{1,2} \\
  c_{2,1} & c_{2,2}
 \end{bmatrix}
 }_{\ve{C}_{\mathrm{b},i}}
 \begin{bmatrix}
  z_{\mathrm{b},1}^{\mathrm{in}} \\ z_{\mathrm{b},2}^{\mathrm{in}}
 \end{bmatrix} \, ,
\end{equation}
i.e., if we retain the coding coefficients and reverse the operations at an 
intermediate node $\node_i$, the coefficient matrix $\ve{C}_{\mathrm{b},i}$ for 
the backward direction is the transpose of the coefficient matrix 
$\ve{C}_i$ for the forward direction
\begin{equation}
 \ve{C}_{\mathrm{b},i} = \ve{C}_i^\T \, .
\end{equation}
The consequence for the individual network channel matrix is stated in the 
following theorem.
\begin{myTheo} \label{Theo:FBDuality}
 The individual network channel matrix $\ve{A}_\mathrm{b}$ for the backward 
 direction in networks which apply LNC is equal to the transpose of the 
 network channel matrix $\ve{A}$ for the forward direction
 \begin{equation}
  \ve{A}_\mathrm{b} = \ve{A}^\T \, ,
 \end{equation}
 given that the coding coefficients are retained, and the operations at the 
 intermediate nodes are reversed.
\end{myTheo}
\begin{proof}
 Consider a layered Variant~\RM{1} representation of an arbitrary network which 
 applies LNC. We first investigate the effects of the reversion of the 
 communication direction on the inter-layer matrices. For that, consider the 
 two adjacent layers depicted in Fig.~\ref{Fig:2layFB}(a).
 \begin{figure}
  \centering
  \psfrag{n11}[cc][bl][1][0]{$\node_1$}
  \psfrag{n12}[cc][bl][1][0]{$\node_2$}
  \psfrag{n13}[cc][bl][1][0]{$\node_3$}
  \psfrag{n21}[cc][bl][1][0]{$\node_A$}
  \psfrag{n22}[cc][bl][1][0]{$\node_B$}
  \psfrag{n23}[cc][bl][1][0]{$\node_C$}
  \psfrag{n24}[cc][bl][1][0]{$\node_D$}
  \psfrag{a1}[c][c][.9][0]{$\color{red}c_{A,1}$}
  \psfrag{a2}[c][c][.9][0]{$\color{red}c_{A,2}$}
  \psfrag{a3}[c][c][.9][0]{$\color{red}c_{B,1}$}
  \psfrag{a4}[c][c][.9][0]{$\color{red}c_{B,2}$}
  \psfrag{a5}[c][c][.9][0]{$\color{red}c_{B,3}$}
  \psfrag{a6}[c][c][.9][0]{$\color{red}c_{C,2}$}
  \psfrag{a7}[c][c][.9][0]{$\color{red}c_{C,3}$}
  \psfrag{a8}[c][c][.9][0]{$\color{red}c_{D,3}$}
  \psfrag{l}[c][c][1][0]{$l$}
  \psfrag{l1}[c][c][1][0]{$l+1$}
  \psfrag{A}[c][c][1][0]{$\ve{A}_{l+1,l}$}
  \psfrag{Ab}[c][c][1][0]{$\ve{A}_{l,l+1}$}
  \psfrag{a}[c][c][1][0]{(a)}
  \psfrag{b}[c][c][1][0]{(b)}
  \includegraphics[width=\columnwidth]{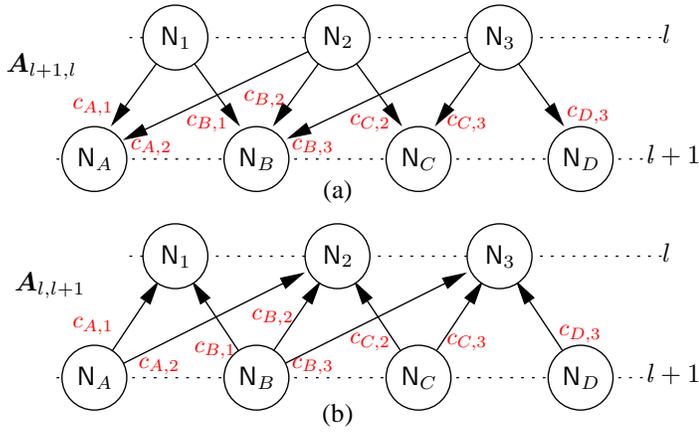}
  \caption{Two exemplary layers of a communication network in forward (a) and 
           in backward (b) direction.}
  \label{Fig:2layFB}           
 \end{figure}
 The inter-layer matrix for the forward direction $\ve{A}_{l+1,l}$ results in
 \begin{equation}
  \ve{A}_{l+1,l} = 
  \begin{bmatrix}
   c_{A,1} & c_{A,2} & 0 \\
   c_{B,1} & c_{B,2} & c_{B,3} \\
   0 & c_{C,2} & c_{C,3} \\
   0 & 0 & c_{D,3}
  \end{bmatrix} \, .
 \end{equation}
 If we reverse the processing at the nodes as described above and retain the 
 coding coefficients, the coefficient which corresponded to edge 
 $e_{i,j}$, now corresponds to the reversed edge $e_{j,i}$, cf. 
 Fig.~\ref{Fig:2layFB}(b). Due to the fact that the roles of the layers are 
 interchanged (i.e., the ``transmitting'' layer is now the ``receiving'' layer
 and vice versa) the inter-layer matrix for the backward direction 
 $\ve{A}_{l,l+1}$ is the transposed version of the one for the forward 
 direction
 \begin{equation}
  \ve{A}_{l,l+1} = 
  \begin{bmatrix}
   c_{A,1} & c_{B,1} & 0 & 0 \\
   c_{A,2} & c_{B,2} & c_{C,2} & 0 \\
   0 & c_{B,3} & c_{C,3} & c_{D,3}
  \end{bmatrix} = 
  \ve{A}_{l+1,l}^\T \, .
 \end{equation}
 Inserting this into (\ref{Eq:Aprod}) yields
 \begin{eqnarray}
  \ve{A}_\mathrm{b} &=& \ve{A}_{1,2} \cdot \ve{A}_{2,3} \cdots \ve{A}_{L-1,L}
  \nonumber\\
  &=& \ve{A}_{2,1}^\T \cdot \ve{A}_{3,2}^\T \cdots \ve{A}_{L,L-1}^\T
  \nonumber\\
  &=& (\ve{A}_{L,L-1} \cdots \ve{A}_{3,2} \cdot \ve{A}_{2,1})^\T
  \nonumber\\
  &=& \ve{A}^\T \, .
 \end{eqnarray}
\end{proof}
\vspace*{-3mm}
Theorem~\ref{Theo:FBDuality} can be seen as an analogon to the famous 
\emph{uplink-downlink duality} from MIMO communications, e.g., 
\cite{UDDuality}, which states that the channel matrix $\ve{H}_\mathrm{u}$ for 
the uplink is equal to the Hermitian transpose of the channel matrix 
$\ve{H}_\mathrm{d}$ for the downlink, i.e., 
$\ve{H}_\mathrm{u}^{} = \ve{H}_\mathrm{d}^\H$ (in the complex baseband).

If $N_k > n$, the ``reverse source'' node $\Dnode_k$ has more outgoing edges 
than the ``reverse destination'' node $\Snode$ incoming ones. As a consequence, 
the ``reverse source'' $\Dnode_k$ cannot simply transmit $N_k$ individual 
transmit symbols. Rather, we have to force $\ve{A}$ to be square, by selecting 
$n$ linearly independent rows and deleting the remaining $N_k - n$ ones. In the 
graph $\Graph$ this corresponds to deleting the corresponding $N_k - n$ 
incoming edges of $\Dnode_k$. Another possibility to resolve the problem of 
having too many outgoing edges at the ``reverse source'' is the application of 
\emph{precoding} \cite{Fis02}, which is denoted as \emph{coding at the source} 
\cite{FrS07} in the context of NC. Then, the ``reverse source'' transmits $n$ 
individual transmit symbols and $N_k - n$ linear combinations of them.

The consequence of Theorem~\ref{Theo:FBDuality} on the validity of the linear
network code is as follows.
\vspace*{-2mm}
\begin{myTheo}
 If a linear network code for the forward direction is valid, then it is also 
 valid for the backward direction.
\end{myTheo}
\vspace*{-3mm}
\begin{proof}
 A linear network code is valid, if the network channel matrix has a rank 
 equal to $n$. If this is given, then the rank of the network channel matrix
 for the backward direction is also equal to $n$
 \begin{equation}
  \rank{\ve{A}_\mathrm{b}} = \rank{\ve{A}^\T} = \rank{\ve{A}} = n \, .
 \end{equation}
\end{proof}
\vspace*{-3mm}
Thus, in a bidirectional NC scenario it is sufficient to design a linear 
network code for one direction, e.g., with the aid of the \emph{linear 
information flow (LIF) algorithm} \cite{jaggy2005}. This code can then be used 
also for the backward direction if the operations at the intermediate nodes 
are reversed according to Fig.~\ref{Fig:revert}.

\vspace*{-2mm}
\section{Conclusion} \label{Sec:Conclusion}
\noindent In this work, we have classified NC variants, and have shown that all 
variants can be traced back to the most basic one---NC Variant~\RM{1}. We have 
studied layered networks, and the application of LNC to such networks. 
Moreover, a technique called layering has been proposed, which allows us to 
introduce a layered structure into arbitrary, non-layered networks. With the 
aid of the layered Variant~\RM{1} representation of communication networks we 
were able to state an algebraic expression of the mincut, and to derive the 
forward-backward duality for LNC, which can be seen as an analogon to the 
famous uplink-downlink duality for MIMO channels \cite{UDDuality}. Furthermore, 
we already exploited the advantages of the layered Variant~\RM{1} 
representation of a communication network in the context of random linear 
network coding in \cite{schotsch2015pout}.

\vspace*{-2mm}


\end{document}